\RequirePackage{fix-cm}
\documentclass[10pt,letterpaper]{article}
\usepackage{setspace} 
\usepackage{mathptmx}
\usepackage{amssymb}
\usepackage{amsthm}
\usepackage{amsmath}
\usepackage{graphicx}
\usepackage[table]{xcolor}
\usepackage{multirow}
\theoremstyle{definition}
\newtheorem{definition}{Definition}[section]

\newtheorem{lemma}{Lemma}[section]

\newtheorem{example}{Example}[section]

\newtheorem{remark}{Remark}[section]
\begin{document}
\pagenumbering{arabic}
%
%
\begin{center}
{\Large\bf{Consciousness and the structuring property of typical data}}\\
Jonathan W. D. Mason, jwmason@jwmason.net, See Published OnlineOpen\\
\begin{spacing}{1}
Complexity, doi: 10.1002/cplx.21431, Copyright \copyright 2012 Wiley Periodicals, Inc.
\end{spacing}
\end{center}
\vspace{-1cm}
\begin{abstract}
The theoretical base for consciousness, in particular an explanation of how consciousness is defined by the brain, has long been sought by science. We propose a partial theory of consciousness as relations defined by typical data. The theory is based on the idea that a brain state on its own is almost meaningless but in the context of the typical brain states, defined by the brain's structure, a particular brain state is highly structured by relations. The proposed theory can be applied and tested both theoretically and experimentally. Precisely how typical data determines relations is fully established using discrete mathematics.
\end{abstract}
%
%
\section{Introduction}
\label{intro}
In neuroscience the neural correlates of consciousness provide an important empirical base for consciousness but not a theoretical one. To clarify, a theoretical base is a predictive theory that is free from empirical methodology whilst usually appealing to, and revealing aspects of, the innate mathematical properties of what is being studied. In contrast the neural correlates of consciousness at some stage rely on obtaining information about a person's experience by asking them or by considering their sensory input. Subsequently a given experience can be associated with the aspects of a person's neurological state that are always observed for that experience. To exemplify the difference, compare Newtonian mechanics with astronomical predictions based on astronomical tables. Importantly it is expected that the neural correlates of consciousness alone cannot provide a satisfactory explanation of consciousness since this would invoke some unknown agency that can discover the external cause of a particular neurological state within the brain so as to associate that state with an appropriate experience. Hence an important requirement of a theoretical base for consciousness is that it should avoid the use of any prior knowledge of what stimulates the senses. We should expect the brain itself to fully define conscious experience all be it having been stimulated by the senses. To assess whether a particular theory meets this requirement we also need a clear notion of what consciousness is. Whilst consensus in this regard will be hard to come by, it can be argued that one fundamental aspect of consciousness is the role played by relations such as those that define geometric content or the individuality of objects, their relationships and type such as visual or auditory. We therefore postulate that our conscious experience could largely be a mathematical structure defined by relations. In this case the principle underlying how the brain simultaneously defines all the required relations is needed. For example, since the part of conscious experience that correlates with the state of the primary visual cortex is of a metric space viewed from a particular position, we expect that the primary visual cortex ought to define relations between neurons, or other identifiable nodes, that result in a metric space. This paper proposes a theory that may satisfy these requirements whilst being theoretically and experimentally amenable to the scientific method. Of course the scientific literature does already include important contributions towards establishing a theoretical base for consciousness. Perhaps the most prominent of these is the theory of consciousness as integrated information proposed by Giulio Tononi in 2008, see \cite{Tononi}. Tononi had previously worked with Gerald Edelman the Nobel Prize-winning immunologist and subsequent neuroscientist. Together they wrote a book entitled A Universe of Consciousness, see \cite{Edelman}, which provides significant scientific insight towards an account of consciousness. However, whilst the importance of relations is evident in their work, their emphasis does not suggest how the content of consciousness might be defined by the brain. A review of the book by Giorgio Ascoli, see \cite{Ascoli}, points out that the authors focus on the properties of the neural process such as integrated activity in the highly reentrant dynamic core, where the dynamic core is a large part of the thalamocortical system, and also on the properties of consciousness such as unity, privateness, coherence and informativeness. In Ascoli's view the book does not address the question of why a sensation corresponds to a specific state of the dynamic core as opposed to another one. In this respect I support the view that the relations defined by the brain are important. It can be seen from \cite{Shinkareva} that the brain defines relationships between certain patterns of activity occurring in various sensory regions of the brain. For example, for a given pattern of activity in the visual cortex we can ask whether it is typical for another particular pattern of activity to be present at the same time in the auditory cortex. If so then the given pattern is related to the latter pattern. Consider how such a relationship might be contributing to the experience of seeing a picture of Albert Einstein whilst hearing the name Albert as opposed to hearing the noun apple. For now the experience associated with a particular pattern of activity may be known from the neural correlates of consciousness. However the relationships that the brain defines between patterns allows more to be derived about a person's experience than that associated to the patterns in the sensory regions of the brain alone. Hence we should try to move down from this higher semantic level replacing neural correlates of consciousness with derivations involving relations as we go if possible. I do not however doubt the enduring relevance and importance of Edelman and Tononi's work such is the knowledge and insight it provides.\\
The mathematics in this paper is straightforward involving binary relations, matrix tables and a small amount of graph theory. The relevance of such mathematics for the brain has been noticed before particularly in the study of anatomical and functional connectivity, see \cite{Sporns}, which is a different, and yet associated, purpose to that of this paper concerning consciousness.\\
We will start by considering the following properties of the brain that are available for consciousness, noting that the list is not intended to be exhaustive,:
\begin{enumerate}
\item[(i)]
the brain has a large number of identifiable nodes by which we mean neurons in this paper, but more generally possibly cortical columns;
\item[(ii)]
the brain is capable of a large number of states where a brain state is a possible and probable aggregate state of all the brain's nodes;
\item[(iii)]
to some extent there is some type of ordering on the collection of brain states since the brain has some of the properties of an endofunction, all be it under perturbation by the senses.
\end{enumerate}
In this paper we will mainly be considering (i) and (ii) of the above. In this respect Definition \ref{def:CSTD} will be useful where, when applied to the brain, the elements of $S$ are the neurons. Merely to keep things simple we will mainly restrict our selves to nodes that have a two state repertoire.
\begin{definition}
\label{def:CSTD}
Let $S$ be a nonempty finite set, $n:=\# S$. Then a set, for an arbitrary index label $i$,
\begin{equation}
S_{i}:=\{(a,f_{i}(a)):a\in S,\quad f_{i}:S\rightarrow\{0,1\}\},\quad\mbox{where $f_{i}$ is a map},
\end{equation}
will be called {\em{a data element}} for $S$. The set of all data elements for $S$ is denoted $\Omega_{S}$ so that $\#\Omega_{S}=2^{n}$. If a particular subset $T\subseteq\Omega_{S}$ has been associated with $S$ then we will call $T$ the {\em{typical data}} for $S$. Further in such cases we will refer to $S$ as the {\em{carrier set}}. An element $S_{i}\in T$ will be called {\em{a typical data element}}.
\end{definition}
Before we consider the brain the following motivating example will be useful.
\begin{example}
\label{exa:DPP}
We will consider what could appropriately be called: {\em{The definitive player problem}}. The purpose of this simple example is to introduce the idea that typical data can define a structure on a carrier set which in tern gives an interpretation of each typical data element. Consider a library of compact discs and suppose that these discs have all been made to a generic template in the sense that the locations of the bits, either 0 or 1, are the same for all discs. Further suppose that the discs all produce highly structured output on some standard player which always reads off the bits in the same order relative to the generic template. In the language of Definition \ref{def:CSTD} the generic template is the carrier set $S$ and the library is the typical data $T$. Now suppose we have two of these discs $S_{1},S_{2}\in T$ where, on the standard player, $S_{1}$ is Beethoven and $S_{2}$ is Elgar. On some nonstandard player where the order in which the bits are read is different to the standard player it could be that $S_{1}$ is Mozart and $S_{2}$ is something else, possibly white noise, depending on the reading order. Therefore a single disk on its own is almost meaningless. However, by requiring highly structured output, each disc $S_{i}$ in the library defines a subset of the set of all players. By taking the intersection of all these subsets we will be left with relatively few players including the standard player. If the library is large enough and we could measure how structured an output is then the typical data might determine a definitive player and hence, in the context of the library, $S_{1}$ is Beethoven and $S_{2}$ is Elgar.\\
The definitive player in this example is essentially a relation between the bits on the generic disc template, i.e. the carrier set, such that almost every bit is related to two other bits so as to form a sequence up to a choice of direction. When a disc from the library is played on the definitive player the output has relatively few abrupt transitions in output frequency and so there is some similarity between the relation on the carrier set and what is written on the discs.\\
We finish this example by mentioning that there are plenty of different choices of typical data, i.e. libraries, available and in particular many more than there are players. If there are $n$ bit locations on the generic disc template, so that $\# S=n$, then there are $n!$ different players by which we mean $n!$ different sequences of these bit locations. Further the number of different discs that can be written is $2^{n}$, that is $\#\Omega_{S}=2^{n}$. Therefore the number of different subsets of $\Omega_{S}$ is $2^{2^{n}}$ and it is straightforward to show by induction that $2^{2^{n}}>n!$ for all $n\in\mathbb{N}$.
\end{example}
In the next section we will see that the appropriate relation to put on the carrier set, if unique, is explicitly determined by the typical data itself. Suppose in Example 1.1 that instead of the data points on the discs having a two state repertoire, bits, there were as many states as output frequencies or that the nodes on the generic disc template are the bytes instead of the bits. Then the theory in the next section would apply to Example \ref{exa:DPP} and there would not be a problem concerning how to measure the quantity of structure of an output. Moreover towards the end of this paper we will argue that the theory presented solves what is known as the binding problem.
\section{Relations defined by typical data}
\label{sec:1}
We will refer to Table \ref{tab:1} below several times in this section.
\begin{table}
\centering
\caption{One dimensional arrangements of four bit data elements.}
\label{tab:1}
\setlength{\tabcolsep}{4pt}
\resizebox{9.5cm}{!}{
\begin{tabular}{c|@{\hspace{1.7mm}}c@{\hspace{1.7mm}}|@{\hspace{1.7mm}}c@{\hspace{1.7mm}}|@{\hspace{1.7mm}}c@{\hspace{1.7mm}}|@{\hspace{1.7mm}}c@{\hspace{1.7mm}}|@{\hspace{1.7mm}}c@{\hspace{1.7mm}}|@{\hspace{1.7mm}}c@{\hspace{1.7mm}}|@{\hspace{1.7mm}}c@{\hspace{1.7mm}}|@{\hspace{1.7mm}}c@{\hspace{1.7mm}}|@{\hspace{1.7mm}}c@{\hspace{1.7mm}}|@{\hspace{1.7mm}}c@{\hspace{1.7mm}}|@{\hspace{1.7mm}}c@{\hspace{1.7mm}}|@{\hspace{1.7mm}}c@{\hspace{1.7mm}}|@{\hspace{1.7mm}}c@{\hspace{1.7mm}}|@{\hspace{1.7mm}}c@{\hspace{1.7mm}}|@{\hspace{1.7mm}}c@{\hspace{1.7mm}}|@{\hspace{1.7mm}}c@{\hspace{1.7mm}}|@{\hspace{1.7mm}}c@{\hspace{1.7mm}}|@{\hspace{1.7mm}}c@{\hspace{1.7mm}}|@{\hspace{1.7mm}}c@{\hspace{1.7mm}}|@{\hspace{1.7mm}}c@{\hspace{1.7mm}}|@{\hspace{1.7mm}}c@{\hspace{1.7mm}}|@{\hspace{1.7mm}}c@{\hspace{1.7mm}}|@{\hspace{1.7mm}}c@{\hspace{1.7mm}}|@{\hspace{1.7mm}}c@{\hspace{1.7mm}}|}
\multirow{5}{*}{}
&$\mathtt{a}$&$\mathtt{a}$&$\mathtt{a}$&\cellcolor[gray]{0.8}$\mathtt{a}$&$\mathtt{a}$&$\mathtt{a}$&$\mathtt{b}$&$\mathtt{b}$&$\mathtt{b}$&$\mathtt{b}$&$\mathtt{b}$&\cellcolor[gray]{0.8}$\mathtt{b}$&$\mathtt{c}$&$\mathtt{c}$&$\mathtt{c}$&$\mathtt{c}$&$\mathtt{c}$&$\mathtt{c}$&$\mathtt{d}$&$\mathtt{d}$&$\mathtt{d}$&$\mathtt{d}$&$\mathtt{d}$&$\mathtt{d}$\\
&$\mathtt{b}$&$\mathtt{b}$&$\mathtt{c}$&\cellcolor[gray]{0.8}$\mathtt{c}$&$\mathtt{d}$&$\mathtt{d}$&$\mathtt{a}$&$\mathtt{a}$&$\mathtt{c}$&$\mathtt{c}$&$\mathtt{d}$&\cellcolor[gray]{0.8}$\mathtt{d}$&$\mathtt{a}$&$\mathtt{a}$&$\mathtt{b}$&$\mathtt{b}$&$\mathtt{d}$&$\mathtt{d}$&$\mathtt{a}$&$\mathtt{a}$&$\mathtt{b}$&$\mathtt{b}$&$\mathtt{c}$&$\mathtt{c}$\\
&$\mathtt{c}$&$\mathtt{d}$&$\mathtt{b}$&\cellcolor[gray]{0.8}$\mathtt{d}$&$\mathtt{b}$&$\mathtt{c}$&$\mathtt{c}$&$\mathtt{d}$&$\mathtt{a}$&$\mathtt{d}$&$\mathtt{a}$&\cellcolor[gray]{0.8}$\mathtt{c}$&$\mathtt{b}$&$\mathtt{d}$&$\mathtt{a}$&$\mathtt{d}$&$\mathtt{a}$&$\mathtt{b}$&$\mathtt{b}$&$\mathtt{c}$&$\mathtt{a}$&$\mathtt{c}$&$\mathtt{a}$&$\mathtt{b}$\\
&$\mathtt{d}$&$\mathtt{c}$&$\mathtt{d}$&\cellcolor[gray]{0.8}$\mathtt{b}$&$\mathtt{c}$&$\mathtt{b}$&$\mathtt{d}$&$\mathtt{c}$&$\mathtt{d}$&$\mathtt{a}$&$\mathtt{c}$&\cellcolor[gray]{0.8}$\mathtt{a}$&$\mathtt{d}$&$\mathtt{b}$&$\mathtt{d}$&$\mathtt{a}$&$\mathtt{b}$&$\mathtt{a}$&$\mathtt{c}$&$\mathtt{b}$&$\mathtt{c}$&$\mathtt{a}$&$\mathtt{b}$&$\mathtt{a}$\\
&&&&&&&&&&&&&&&&&&&&&&&&\\[-3.7mm]
\hline
\multirow{6}{*}{\large $S_{1}$}
&&&&&&&&&&&&&&&&&&&&&&&&\\[-3.7mm]
&0&0&0&0&0&0&0&0&0&0&0&0&0&0&0&0&0&0&0&0&0&0&0&0\\
&0&0&0&0&0&0&0&0&0&0&0&0&0&0&0&0&0&0&0&0&0&0&0&0\\
&0&0&0&0&0&0&0&0&0&0&0&0&0&0&0&0&0&0&0&0&0&0&0&0\\
&0&0&0&0&0&0&0&0&0&0&0&0&0&0&0&0&0&0&0&0&0&0&0&0\\
&&&&&&&&&&&&&&&&&&&&&&&&\\[-3.7mm]
\hline
\multirow{6}{*}{\large $S_{2}$}
&&&&&&&&&&&&&&&&&&&&&&&&\\[-3.7mm]
&0&0&0&0&0&0&0&0&0&0&0&0&0&0&0&0&0&0&\cellcolor[gray]{0.8}1&\cellcolor[gray]{0.8}1&\cellcolor[gray]{0.8}1&\cellcolor[gray]{0.8}1&\cellcolor[gray]{0.8}1&\cellcolor[gray]{0.8}1\\
&0&0&0&0&\cellcolor[gray]{0.8}1&\cellcolor[gray]{0.8}1&0&0&0&0&\cellcolor[gray]{0.8}1&\cellcolor[gray]{0.8}1&0&0&0&0&\cellcolor[gray]{0.8}1&\cellcolor[gray]{0.8}1&0&0&0&0&0&0\\
&0&\cellcolor[gray]{0.8}1&0&\cellcolor[gray]{0.8}1&0&0&0&\cellcolor[gray]{0.8}1&0&\cellcolor[gray]{0.8}1&0&0&0&\cellcolor[gray]{0.8}1&0&\cellcolor[gray]{0.8}1&0&0&0&0&0&0&0&0\\
&\cellcolor[gray]{0.8}1&0&\cellcolor[gray]{0.8}1&0&0&0&\cellcolor[gray]{0.8}1&0&\cellcolor[gray]{0.8}1&0&0&0&\cellcolor[gray]{0.8}1&0&\cellcolor[gray]{0.8}1&0&0&0&0&0&0&0&0&0\\
&&&&&&&&&&&&&&&&&&&&&&&&\\[-3.7mm]
\hline
\multirow{6}{*}{\large $S_{3}$}
&&&&&&&&&&&&&&&&&&&&&&&&\\[-3.7mm]
&0&0&0&0&0&0&0&0&0&0&0&0&\cellcolor[gray]{0.8}1&\cellcolor[gray]{0.8}1&\cellcolor[gray]{0.8}1&\cellcolor[gray]{0.8}1&\cellcolor[gray]{0.8}1&\cellcolor[gray]{0.8}1&0&0&0&0&0&0\\
&0&0&\cellcolor[gray]{0.8}1&\cellcolor[gray]{0.8}1&0&0&0&0&\cellcolor[gray]{0.8}1&\cellcolor[gray]{0.8}1&0&0&0&0&0&0&0&0&0&0&0&0&\cellcolor[gray]{0.8}1&\cellcolor[gray]{0.8}1\\
&\cellcolor[gray]{0.8}1&0&0&0&0&\cellcolor[gray]{0.8}1&\cellcolor[gray]{0.8}1&0&0&0&0&\cellcolor[gray]{0.8}1&0&0&0&0&0&0&0&\cellcolor[gray]{0.8}1&0&\cellcolor[gray]{0.8}1&0&0\\
&0&\cellcolor[gray]{0.8}1&0&0&\cellcolor[gray]{0.8}1&0&0&\cellcolor[gray]{0.8}1&0&0&\cellcolor[gray]{0.8}1&0&0&0&0&0&0&0&\cellcolor[gray]{0.8}1&0&\cellcolor[gray]{0.8}1&0&0&0\\
&&&&&&&&&&&&&&&&&&&&&&&&\\[-3.7mm]
\hline
\multirow{6}{*}{\large $S_{4}$}
&&&&&&&&&&&&&&&&&&&&&&&&\\[-3.7mm]
&0&0&0&0&0&0&\cellcolor[gray]{0.8}1&\cellcolor[gray]{0.8}1&\cellcolor[gray]{0.8}1&\cellcolor[gray]{0.8}1&\cellcolor[gray]{0.8}1&\cellcolor[gray]{0.8}1&0&0&0&0&0&0&0&0&0&0&0&0\\
&\cellcolor[gray]{0.8}1&\cellcolor[gray]{0.8}1&0&0&0&0&0&0&0&0&0&0&0&0&\cellcolor[gray]{0.8}1&\cellcolor[gray]{0.8}1&0&0&0&0&\cellcolor[gray]{0.8}1&\cellcolor[gray]{0.8}1&0&0\\
&0&0&\cellcolor[gray]{0.8}1&0&\cellcolor[gray]{0.8}1&0&0&0&0&0&0&0&\cellcolor[gray]{0.8}1&0&0&0&0&\cellcolor[gray]{0.8}1&\cellcolor[gray]{0.8}1&0&0&0&0&\cellcolor[gray]{0.8}1\\
&0&0&0&\cellcolor[gray]{0.8}1&0&\cellcolor[gray]{0.8}1&0&0&0&0&0&0&0&\cellcolor[gray]{0.8}1&0&0&\cellcolor[gray]{0.8}1&0&0&\cellcolor[gray]{0.8}1&0&0&\cellcolor[gray]{0.8}1&0\\
&&&&&&&&&&&&&&&&&&&&&&&&\\[-3.7mm]
\hline
\multirow{6}{*}{\large $\underline{S_{5}}$}
&&&&&&&&&&&&&&&&&&&&&&&&\\[-3.7mm]
&\cellcolor[gray]{0.8}1&\cellcolor[gray]{0.8}1&\cellcolor[gray]{0.8}1&\cellcolor[gray]{0.8}1&\cellcolor[gray]{0.8}1&\cellcolor[gray]{0.8}1&0&0&0&0&0&0&0&0&0&0&0&0&0&0&0&0&0&0\\
&0&0&0&0&0&0&\cellcolor[gray]{0.8}1&\cellcolor[gray]{0.8}1&0&0&0&0&\cellcolor[gray]{0.8}1&\cellcolor[gray]{0.8}1&0&0&0&0&\cellcolor[gray]{0.8}1&\cellcolor[gray]{0.8}1&0&0&0&0\\
&0&0&0&0&0&0&0&0&\cellcolor[gray]{0.8}1&0&\cellcolor[gray]{0.8}1&0&0&0&\cellcolor[gray]{0.8}1&0&\cellcolor[gray]{0.8}1&0&0&0&\cellcolor[gray]{0.8}1&0&\cellcolor[gray]{0.8}1&0\\
&0&0&0&0&0&0&0&0&0&\cellcolor[gray]{0.8}1&0&\cellcolor[gray]{0.8}1&0&0&0&\cellcolor[gray]{0.8}1&0&\cellcolor[gray]{0.8}1&0&0&0&\cellcolor[gray]{0.8}1&0&\cellcolor[gray]{0.8}1\\
&&&&&&&&&&&&&&&&&&&&&&&&\\[-3.7mm]
\hline
\multirow{6}{*}{\large $S_{6}$}
&&&&&&&&&&&&&&&&&&&&&&&&\\[-3.7mm]
&0&0&0&0&0&0&0&0&0&0&0&0&\cellcolor[gray]{0.8}1&\cellcolor[gray]{0.8}1&\cellcolor[gray]{0.8}1&\cellcolor[gray]{0.8}1&\cellcolor[gray]{0.8}1&\cellcolor[gray]{0.8}1&\cellcolor[gray]{0.8}1&\cellcolor[gray]{0.8}1&\cellcolor[gray]{0.8}1&\cellcolor[gray]{0.8}1&\cellcolor[gray]{0.8}1&\cellcolor[gray]{0.8}1\\
&0&0&\cellcolor[gray]{0.8}1&\cellcolor[gray]{0.8}1&\cellcolor[gray]{0.8}1&\cellcolor[gray]{0.8}1&0&0&\cellcolor[gray]{0.8}1&\cellcolor[gray]{0.8}1&\cellcolor[gray]{0.8}1&\cellcolor[gray]{0.8}1&0&0&0&0&\cellcolor[gray]{0.8}1&\cellcolor[gray]{0.8}1&0&0&0&0&\cellcolor[gray]{0.8}1&\cellcolor[gray]{0.8}1\\
&\cellcolor[gray]{0.8}1&\cellcolor[gray]{0.8}1&0&\cellcolor[gray]{0.8}1&0&\cellcolor[gray]{0.8}1&\cellcolor[gray]{0.8}1&\cellcolor[gray]{0.8}1&0&\cellcolor[gray]{0.8}1&0&\cellcolor[gray]{0.8}1&0&\cellcolor[gray]{0.8}1&0&\cellcolor[gray]{0.8}1&0&0&0&\cellcolor[gray]{0.8}1&0&\cellcolor[gray]{0.8}1&0&0\\
&\cellcolor[gray]{0.8}1&\cellcolor[gray]{0.8}1&\cellcolor[gray]{0.8}1&0&\cellcolor[gray]{0.8}1&0&\cellcolor[gray]{0.8}1&\cellcolor[gray]{0.8}1&\cellcolor[gray]{0.8}1&0&\cellcolor[gray]{0.8}1&0&\cellcolor[gray]{0.8}1&0&\cellcolor[gray]{0.8}1&0&0&0&\cellcolor[gray]{0.8}1&0&\cellcolor[gray]{0.8}1&0&0&0\\
&&&&&&&&&&&&&&&&&&&&&&&&\\[-3.7mm]
\hline
\multirow{6}{*}{\large $S_{7}$}
&&&&&&&&&&&&&&&&&&&&&&&&\\[-3.7mm]
&0&0&0&0&0&0&\cellcolor[gray]{0.8}1&\cellcolor[gray]{0.8}1&\cellcolor[gray]{0.8}1&\cellcolor[gray]{0.8}1&\cellcolor[gray]{0.8}1&\cellcolor[gray]{0.8}1&0&0&0&0&0&0&\cellcolor[gray]{0.8}1&\cellcolor[gray]{0.8}1&\cellcolor[gray]{0.8}1&\cellcolor[gray]{0.8}1&\cellcolor[gray]{0.8}1&\cellcolor[gray]{0.8}1\\
&\cellcolor[gray]{0.8}1&\cellcolor[gray]{0.8}1&0&0&\cellcolor[gray]{0.8}1&\cellcolor[gray]{0.8}1&0&0&0&0&\cellcolor[gray]{0.8}1&\cellcolor[gray]{0.8}1&0&0&\cellcolor[gray]{0.8}1&\cellcolor[gray]{0.8}1&\cellcolor[gray]{0.8}1&\cellcolor[gray]{0.8}1&0&0&\cellcolor[gray]{0.8}1&\cellcolor[gray]{0.8}1&0&0\\
&0&\cellcolor[gray]{0.8}1&\cellcolor[gray]{0.8}1&\cellcolor[gray]{0.8}1&\cellcolor[gray]{0.8}1&0&0&\cellcolor[gray]{0.8}1&0&\cellcolor[gray]{0.8}1&0&0&\cellcolor[gray]{0.8}1&\cellcolor[gray]{0.8}1&0&\cellcolor[gray]{0.8}1&0&\cellcolor[gray]{0.8}1&\cellcolor[gray]{0.8}1&0&0&0&0&\cellcolor[gray]{0.8}1\\
&\cellcolor[gray]{0.8}1&0&\cellcolor[gray]{0.8}1&\cellcolor[gray]{0.8}1&0&\cellcolor[gray]{0.8}1&\cellcolor[gray]{0.8}1&0&\cellcolor[gray]{0.8}1&0&0&0&\cellcolor[gray]{0.8}1&\cellcolor[gray]{0.8}1&\cellcolor[gray]{0.8}1&0&\cellcolor[gray]{0.8}1&0&0&\cellcolor[gray]{0.8}1&0&0&\cellcolor[gray]{0.8}1&0\\
&&&&&&&&&&&&&&&&&&&&&&&&\\[-3.7mm]
\hline
\multirow{6}{*}{\large $S_{8}$}
&&&&&&&&&&&&&&&&&&&&&&&&\\[-3.7mm]
&\cellcolor[gray]{0.8}1&\cellcolor[gray]{0.8}1&\cellcolor[gray]{0.8}1&\cellcolor[gray]{0.8}1&\cellcolor[gray]{0.8}1&\cellcolor[gray]{0.8}1&0&0&0&0&0&0&0&0&0&0&0&0&\cellcolor[gray]{0.8}1&\cellcolor[gray]{0.8}1&\cellcolor[gray]{0.8}1&\cellcolor[gray]{0.8}1&\cellcolor[gray]{0.8}1&\cellcolor[gray]{0.8}1\\
&0&0&0&0&\cellcolor[gray]{0.8}1&\cellcolor[gray]{0.8}1&\cellcolor[gray]{0.8}1&\cellcolor[gray]{0.8}1&0&0&\cellcolor[gray]{0.8}1&\cellcolor[gray]{0.8}1&\cellcolor[gray]{0.8}1&\cellcolor[gray]{0.8}1&0&0&\cellcolor[gray]{0.8}1&\cellcolor[gray]{0.8}1&\cellcolor[gray]{0.8}1&\cellcolor[gray]{0.8}1&0&0&0&0\\
&0&\cellcolor[gray]{0.8}1&0&\cellcolor[gray]{0.8}1&0&0&0&\cellcolor[gray]{0.8}1&\cellcolor[gray]{0.8}1&\cellcolor[gray]{0.8}1&\cellcolor[gray]{0.8}1&0&0&\cellcolor[gray]{0.8}1&\cellcolor[gray]{0.8}1&\cellcolor[gray]{0.8}1&\cellcolor[gray]{0.8}1&0&0&0&\cellcolor[gray]{0.8}1&0&\cellcolor[gray]{0.8}1&0\\
&\cellcolor[gray]{0.8}1&0&\cellcolor[gray]{0.8}1&0&0&0&\cellcolor[gray]{0.8}1&0&\cellcolor[gray]{0.8}1&\cellcolor[gray]{0.8}1&0&\cellcolor[gray]{0.8}1&\cellcolor[gray]{0.8}1&0&\cellcolor[gray]{0.8}1&\cellcolor[gray]{0.8}1&0&\cellcolor[gray]{0.8}1&0&0&0&\cellcolor[gray]{0.8}1&0&\cellcolor[gray]{0.8}1\\
&&&&&&&&&&&&&&&&&&&&&&&&\\[-3.7mm]
\hline
\multirow{6}{*}{\large $S_{9}$}
&&&&&&&&&&&&&&&&&&&&&&&&\\[-3.7mm]
&0&0&0&0&0&0&\cellcolor[gray]{0.8}1&\cellcolor[gray]{0.8}1&\cellcolor[gray]{0.8}1&\cellcolor[gray]{0.8}1&\cellcolor[gray]{0.8}1&\cellcolor[gray]{0.8}1&\cellcolor[gray]{0.8}1&\cellcolor[gray]{0.8}1&\cellcolor[gray]{0.8}1&\cellcolor[gray]{0.8}1&\cellcolor[gray]{0.8}1&\cellcolor[gray]{0.8}1&0&0&0&0&0&0\\
&\cellcolor[gray]{0.8}1&\cellcolor[gray]{0.8}1&\cellcolor[gray]{0.8}1&\cellcolor[gray]{0.8}1&0&0&0&0&\cellcolor[gray]{0.8}1&\cellcolor[gray]{0.8}1&0&0&0&0&\cellcolor[gray]{0.8}1&\cellcolor[gray]{0.8}1&0&0&0&0&\cellcolor[gray]{0.8}1&\cellcolor[gray]{0.8}1&\cellcolor[gray]{0.8}1&\cellcolor[gray]{0.8}1\\
&\cellcolor[gray]{0.8}1&0&\cellcolor[gray]{0.8}1&0&\cellcolor[gray]{0.8}1&\cellcolor[gray]{0.8}1&\cellcolor[gray]{0.8}1&0&0&0&0&\cellcolor[gray]{0.8}1&\cellcolor[gray]{0.8}1&0&0&0&0&\cellcolor[gray]{0.8}1&\cellcolor[gray]{0.8}1&\cellcolor[gray]{0.8}1&0&\cellcolor[gray]{0.8}1&0&\cellcolor[gray]{0.8}1\\
&0&\cellcolor[gray]{0.8}1&0&\cellcolor[gray]{0.8}1&\cellcolor[gray]{0.8}1&\cellcolor[gray]{0.8}1&0&\cellcolor[gray]{0.8}1&0&0&\cellcolor[gray]{0.8}1&0&0&\cellcolor[gray]{0.8}1&0&0&\cellcolor[gray]{0.8}1&0&\cellcolor[gray]{0.8}1&\cellcolor[gray]{0.8}1&\cellcolor[gray]{0.8}1&0&\cellcolor[gray]{0.8}1&0\\
&&&&&&&&&&&&&&&&&&&&&&&&\\[-3.7mm]
\hline
\multirow{6}{*}{\large $\underline{S_{10}}$}
&&&&&&&&&&&&&&&&&&&&&&&&\\[-3.7mm]
&\cellcolor[gray]{0.8}1&\cellcolor[gray]{0.8}1&\cellcolor[gray]{0.8}1&\cellcolor[gray]{0.8}1&\cellcolor[gray]{0.8}1&\cellcolor[gray]{0.8}1&0&0&0&0&0&0&\cellcolor[gray]{0.8}1&\cellcolor[gray]{0.8}1&\cellcolor[gray]{0.8}1&\cellcolor[gray]{0.8}1&\cellcolor[gray]{0.8}1&\cellcolor[gray]{0.8}1&0&0&0&0&0&0\\
&0&0&\cellcolor[gray]{0.8}1&\cellcolor[gray]{0.8}1&0&0&\cellcolor[gray]{0.8}1&\cellcolor[gray]{0.8}1&\cellcolor[gray]{0.8}1&\cellcolor[gray]{0.8}1&0&0&\cellcolor[gray]{0.8}1&\cellcolor[gray]{0.8}1&0&0&0&0&\cellcolor[gray]{0.8}1&\cellcolor[gray]{0.8}1&0&0&\cellcolor[gray]{0.8}1&\cellcolor[gray]{0.8}1\\
&\cellcolor[gray]{0.8}1&0&0&0&0&\cellcolor[gray]{0.8}1&\cellcolor[gray]{0.8}1&0&\cellcolor[gray]{0.8}1&0&\cellcolor[gray]{0.8}1&\cellcolor[gray]{0.8}1&0&0&\cellcolor[gray]{0.8}1&0&\cellcolor[gray]{0.8}1&0&0&\cellcolor[gray]{0.8}1&\cellcolor[gray]{0.8}1&\cellcolor[gray]{0.8}1&\cellcolor[gray]{0.8}1&0\\
&0&\cellcolor[gray]{0.8}1&0&0&\cellcolor[gray]{0.8}1&0&0&\cellcolor[gray]{0.8}1&0&\cellcolor[gray]{0.8}1&\cellcolor[gray]{0.8}1&\cellcolor[gray]{0.8}1&0&0&0&\cellcolor[gray]{0.8}1&0&\cellcolor[gray]{0.8}1&\cellcolor[gray]{0.8}1&0&\cellcolor[gray]{0.8}1&\cellcolor[gray]{0.8}1&0&\cellcolor[gray]{0.8}1\\
&&&&&&&&&&&&&&&&&&&&&&&&\\[-3.7mm]
\hline
\multirow{6}{*}{\large $S_{11}$}
&&&&&&&&&&&&&&&&&&&&&&&&\\[-3.7mm]
&\cellcolor[gray]{0.8}1&\cellcolor[gray]{0.8}1&\cellcolor[gray]{0.8}1&\cellcolor[gray]{0.8}1&\cellcolor[gray]{0.8}1&\cellcolor[gray]{0.8}1&\cellcolor[gray]{0.8}1&\cellcolor[gray]{0.8}1&\cellcolor[gray]{0.8}1&\cellcolor[gray]{0.8}1&\cellcolor[gray]{0.8}1&\cellcolor[gray]{0.8}1&0&0&0&0&0&0&0&0&0&0&0&0\\
&\cellcolor[gray]{0.8}1&\cellcolor[gray]{0.8}1&0&0&0&0&\cellcolor[gray]{0.8}1&\cellcolor[gray]{0.8}1&0&0&0&0&\cellcolor[gray]{0.8}1&\cellcolor[gray]{0.8}1&\cellcolor[gray]{0.8}1&\cellcolor[gray]{0.8}1&0&0&\cellcolor[gray]{0.8}1&\cellcolor[gray]{0.8}1&\cellcolor[gray]{0.8}1&\cellcolor[gray]{0.8}1&0&0\\
&0&0&\cellcolor[gray]{0.8}1&0&\cellcolor[gray]{0.8}1&0&0&0&\cellcolor[gray]{0.8}1&0&\cellcolor[gray]{0.8}1&0&\cellcolor[gray]{0.8}1&0&\cellcolor[gray]{0.8}1&0&\cellcolor[gray]{0.8}1&\cellcolor[gray]{0.8}1&\cellcolor[gray]{0.8}1&0&\cellcolor[gray]{0.8}1&0&\cellcolor[gray]{0.8}1&\cellcolor[gray]{0.8}1\\
&0&0&0&\cellcolor[gray]{0.8}1&0&\cellcolor[gray]{0.8}1&0&0&0&\cellcolor[gray]{0.8}1&0&\cellcolor[gray]{0.8}1&0&\cellcolor[gray]{0.8}1&0&\cellcolor[gray]{0.8}1&\cellcolor[gray]{0.8}1&\cellcolor[gray]{0.8}1&0&\cellcolor[gray]{0.8}1&0&\cellcolor[gray]{0.8}1&\cellcolor[gray]{0.8}1&\cellcolor[gray]{0.8}1\\
&&&&&&&&&&&&&&&&&&&&&&&&\\[-3.7mm]
\hline
\multirow{6}{*}{\large $S_{12}$}
&&&&&&&&&&&&&&&&&&&&&&&&\\[-3.7mm]
&0&0&0&0&0&0&\cellcolor[gray]{0.8}1&\cellcolor[gray]{0.8}1&\cellcolor[gray]{0.8}1&\cellcolor[gray]{0.8}1&\cellcolor[gray]{0.8}1&\cellcolor[gray]{0.8}1&\cellcolor[gray]{0.8}1&\cellcolor[gray]{0.8}1&\cellcolor[gray]{0.8}1&\cellcolor[gray]{0.8}1&\cellcolor[gray]{0.8}1&\cellcolor[gray]{0.8}1&\cellcolor[gray]{0.8}1&\cellcolor[gray]{0.8}1&\cellcolor[gray]{0.8}1&\cellcolor[gray]{0.8}1&\cellcolor[gray]{0.8}1&\cellcolor[gray]{0.8}1\\
&\cellcolor[gray]{0.8}1&\cellcolor[gray]{0.8}1&\cellcolor[gray]{0.8}1&\cellcolor[gray]{0.8}1&\cellcolor[gray]{0.8}1&\cellcolor[gray]{0.8}1&0&0&\cellcolor[gray]{0.8}1&\cellcolor[gray]{0.8}1&\cellcolor[gray]{0.8}1&\cellcolor[gray]{0.8}1&0&0&\cellcolor[gray]{0.8}1&\cellcolor[gray]{0.8}1&\cellcolor[gray]{0.8}1&\cellcolor[gray]{0.8}1&0&0&\cellcolor[gray]{0.8}1&\cellcolor[gray]{0.8}1&\cellcolor[gray]{0.8}1&\cellcolor[gray]{0.8}1\\
&\cellcolor[gray]{0.8}1&\cellcolor[gray]{0.8}1&\cellcolor[gray]{0.8}1&\cellcolor[gray]{0.8}1&\cellcolor[gray]{0.8}1&\cellcolor[gray]{0.8}1&\cellcolor[gray]{0.8}1&\cellcolor[gray]{0.8}1&0&\cellcolor[gray]{0.8}1&0&\cellcolor[gray]{0.8}1&\cellcolor[gray]{0.8}1&\cellcolor[gray]{0.8}1&0&\cellcolor[gray]{0.8}1&0&\cellcolor[gray]{0.8}1&\cellcolor[gray]{0.8}1&\cellcolor[gray]{0.8}1&0&\cellcolor[gray]{0.8}1&0&\cellcolor[gray]{0.8}1\\
&\cellcolor[gray]{0.8}1&\cellcolor[gray]{0.8}1&\cellcolor[gray]{0.8}1&\cellcolor[gray]{0.8}1&\cellcolor[gray]{0.8}1&\cellcolor[gray]{0.8}1&\cellcolor[gray]{0.8}1&\cellcolor[gray]{0.8}1&\cellcolor[gray]{0.8}1&0&\cellcolor[gray]{0.8}1&0&\cellcolor[gray]{0.8}1&\cellcolor[gray]{0.8}1&\cellcolor[gray]{0.8}1&0&\cellcolor[gray]{0.8}1&0&\cellcolor[gray]{0.8}1&\cellcolor[gray]{0.8}1&\cellcolor[gray]{0.8}1&0&\cellcolor[gray]{0.8}1&0\\
&&&&&&&&&&&&&&&&&&&&&&&&\\[-3.7mm]
\hline
\multirow{6}{*}{\large $\underline{S_{13}}$}
&&&&&&&&&&&&&&&&&&&&&&&&\\[-3.7mm]
&\cellcolor[gray]{0.8}1&\cellcolor[gray]{0.8}1&\cellcolor[gray]{0.8}1&\cellcolor[gray]{0.8}1&\cellcolor[gray]{0.8}1&\cellcolor[gray]{0.8}1&0&0&0&0&0&0&\cellcolor[gray]{0.8}1&\cellcolor[gray]{0.8}1&\cellcolor[gray]{0.8}1&\cellcolor[gray]{0.8}1&\cellcolor[gray]{0.8}1&\cellcolor[gray]{0.8}1&\cellcolor[gray]{0.8}1&\cellcolor[gray]{0.8}1&\cellcolor[gray]{0.8}1&\cellcolor[gray]{0.8}1&\cellcolor[gray]{0.8}1&\cellcolor[gray]{0.8}1\\
&0&0&\cellcolor[gray]{0.8}1&\cellcolor[gray]{0.8}1&\cellcolor[gray]{0.8}1&\cellcolor[gray]{0.8}1&\cellcolor[gray]{0.8}1&\cellcolor[gray]{0.8}1&\cellcolor[gray]{0.8}1&\cellcolor[gray]{0.8}1&\cellcolor[gray]{0.8}1&\cellcolor[gray]{0.8}1&\cellcolor[gray]{0.8}1&\cellcolor[gray]{0.8}1&0&0&\cellcolor[gray]{0.8}1&\cellcolor[gray]{0.8}1&\cellcolor[gray]{0.8}1&\cellcolor[gray]{0.8}1&0&0&\cellcolor[gray]{0.8}1&\cellcolor[gray]{0.8}1\\
&\cellcolor[gray]{0.8}1&\cellcolor[gray]{0.8}1&0&\cellcolor[gray]{0.8}1&0&\cellcolor[gray]{0.8}1&\cellcolor[gray]{0.8}1&\cellcolor[gray]{0.8}1&\cellcolor[gray]{0.8}1&\cellcolor[gray]{0.8}1&\cellcolor[gray]{0.8}1&\cellcolor[gray]{0.8}1&0&\cellcolor[gray]{0.8}1&\cellcolor[gray]{0.8}1&\cellcolor[gray]{0.8}1&\cellcolor[gray]{0.8}1&0&0&\cellcolor[gray]{0.8}1&\cellcolor[gray]{0.8}1&\cellcolor[gray]{0.8}1&\cellcolor[gray]{0.8}1&0\\
&\cellcolor[gray]{0.8}1&\cellcolor[gray]{0.8}1&\cellcolor[gray]{0.8}1&0&\cellcolor[gray]{0.8}1&0&\cellcolor[gray]{0.8}1&\cellcolor[gray]{0.8}1&\cellcolor[gray]{0.8}1&\cellcolor[gray]{0.8}1&\cellcolor[gray]{0.8}1&\cellcolor[gray]{0.8}1&\cellcolor[gray]{0.8}1&0&\cellcolor[gray]{0.8}1&\cellcolor[gray]{0.8}1&0&\cellcolor[gray]{0.8}1&\cellcolor[gray]{0.8}1&0&\cellcolor[gray]{0.8}1&\cellcolor[gray]{0.8}1&0&\cellcolor[gray]{0.8}1\\
&&&&&&&&&&&&&&&&&&&&&&&&\\[-3.7mm]
\hline
\multirow{6}{*}{\large $S_{14}$}
&&&&&&&&&&&&&&&&&&&&&&&&\\[-3.7mm]
&\cellcolor[gray]{0.8}1&\cellcolor[gray]{0.8}1&\cellcolor[gray]{0.8}1&\cellcolor[gray]{0.8}1&\cellcolor[gray]{0.8}1&\cellcolor[gray]{0.8}1&\cellcolor[gray]{0.8}1&\cellcolor[gray]{0.8}1&\cellcolor[gray]{0.8}1&\cellcolor[gray]{0.8}1&\cellcolor[gray]{0.8}1&\cellcolor[gray]{0.8}1&0&0&0&0&0&0&\cellcolor[gray]{0.8}1&\cellcolor[gray]{0.8}1&\cellcolor[gray]{0.8}1&\cellcolor[gray]{0.8}1&\cellcolor[gray]{0.8}1&\cellcolor[gray]{0.8}1\\
&\cellcolor[gray]{0.8}1&\cellcolor[gray]{0.8}1&0&0&\cellcolor[gray]{0.8}1&\cellcolor[gray]{0.8}1&\cellcolor[gray]{0.8}1&\cellcolor[gray]{0.8}1&0&0&\cellcolor[gray]{0.8}1&\cellcolor[gray]{0.8}1&\cellcolor[gray]{0.8}1&\cellcolor[gray]{0.8}1&\cellcolor[gray]{0.8}1&\cellcolor[gray]{0.8}1&\cellcolor[gray]{0.8}1&\cellcolor[gray]{0.8}1&\cellcolor[gray]{0.8}1&\cellcolor[gray]{0.8}1&\cellcolor[gray]{0.8}1&\cellcolor[gray]{0.8}1&0&0\\
&0&\cellcolor[gray]{0.8}1&\cellcolor[gray]{0.8}1&\cellcolor[gray]{0.8}1&\cellcolor[gray]{0.8}1&0&0&\cellcolor[gray]{0.8}1&\cellcolor[gray]{0.8}1&\cellcolor[gray]{0.8}1&\cellcolor[gray]{0.8}1&0&\cellcolor[gray]{0.8}1&\cellcolor[gray]{0.8}1&\cellcolor[gray]{0.8}1&\cellcolor[gray]{0.8}1&\cellcolor[gray]{0.8}1&\cellcolor[gray]{0.8}1&\cellcolor[gray]{0.8}1&0&\cellcolor[gray]{0.8}1&0&\cellcolor[gray]{0.8}1&\cellcolor[gray]{0.8}1\\
&\cellcolor[gray]{0.8}1&0&\cellcolor[gray]{0.8}1&\cellcolor[gray]{0.8}1&0&\cellcolor[gray]{0.8}1&\cellcolor[gray]{0.8}1&0&\cellcolor[gray]{0.8}1&\cellcolor[gray]{0.8}1&0&\cellcolor[gray]{0.8}1&\cellcolor[gray]{0.8}1&\cellcolor[gray]{0.8}1&\cellcolor[gray]{0.8}1&\cellcolor[gray]{0.8}1&\cellcolor[gray]{0.8}1&\cellcolor[gray]{0.8}1&0&\cellcolor[gray]{0.8}1&0&\cellcolor[gray]{0.8}1&\cellcolor[gray]{0.8}1&\cellcolor[gray]{0.8}1\\
&&&&&&&&&&&&&&&&&&&&&&&&\\[-3.7mm]
\hline
\multirow{6}{*}{\large $S_{15}$}
&&&&&&&&&&&&&&&&&&&&&&&&\\[-3.7mm]
&\cellcolor[gray]{0.8}1&\cellcolor[gray]{0.8}1&\cellcolor[gray]{0.8}1&\cellcolor[gray]{0.8}1&\cellcolor[gray]{0.8}1&\cellcolor[gray]{0.8}1&\cellcolor[gray]{0.8}1&\cellcolor[gray]{0.8}1&\cellcolor[gray]{0.8}1&\cellcolor[gray]{0.8}1&\cellcolor[gray]{0.8}1&\cellcolor[gray]{0.8}1&\cellcolor[gray]{0.8}1&\cellcolor[gray]{0.8}1&\cellcolor[gray]{0.8}1&\cellcolor[gray]{0.8}1&\cellcolor[gray]{0.8}1&\cellcolor[gray]{0.8}1&0&0&0&0&0&0\\
&\cellcolor[gray]{0.8}1&\cellcolor[gray]{0.8}1&\cellcolor[gray]{0.8}1&\cellcolor[gray]{0.8}1&0&0&\cellcolor[gray]{0.8}1&\cellcolor[gray]{0.8}1&\cellcolor[gray]{0.8}1&\cellcolor[gray]{0.8}1&0&0&\cellcolor[gray]{0.8}1&\cellcolor[gray]{0.8}1&\cellcolor[gray]{0.8}1&\cellcolor[gray]{0.8}1&0&0&\cellcolor[gray]{0.8}1&\cellcolor[gray]{0.8}1&\cellcolor[gray]{0.8}1&\cellcolor[gray]{0.8}1&\cellcolor[gray]{0.8}1&\cellcolor[gray]{0.8}1\\
&\cellcolor[gray]{0.8}1&0&\cellcolor[gray]{0.8}1&0&\cellcolor[gray]{0.8}1&\cellcolor[gray]{0.8}1&\cellcolor[gray]{0.8}1&0&\cellcolor[gray]{0.8}1&0&\cellcolor[gray]{0.8}1&\cellcolor[gray]{0.8}1&\cellcolor[gray]{0.8}1&0&\cellcolor[gray]{0.8}1&0&\cellcolor[gray]{0.8}1&\cellcolor[gray]{0.8}1&\cellcolor[gray]{0.8}1&\cellcolor[gray]{0.8}1&\cellcolor[gray]{0.8}1&\cellcolor[gray]{0.8}1&\cellcolor[gray]{0.8}1&\cellcolor[gray]{0.8}1\\
&0&\cellcolor[gray]{0.8}1&0&\cellcolor[gray]{0.8}1&\cellcolor[gray]{0.8}1&\cellcolor[gray]{0.8}1&0&\cellcolor[gray]{0.8}1&0&\cellcolor[gray]{0.8}1&\cellcolor[gray]{0.8}1&\cellcolor[gray]{0.8}1&0&\cellcolor[gray]{0.8}1&0&\cellcolor[gray]{0.8}1&\cellcolor[gray]{0.8}1&\cellcolor[gray]{0.8}1&\cellcolor[gray]{0.8}1&\cellcolor[gray]{0.8}1&\cellcolor[gray]{0.8}1&\cellcolor[gray]{0.8}1&\cellcolor[gray]{0.8}1&\cellcolor[gray]{0.8}1\\
&&&&&&&&&&&&&&&&&&&&&&&&\\[-3.7mm]
\hline
\multirow{5}{*}{\large $S_{16}$}
&&&&&&&&&&&&&&&&&&&&&&&&\\[-3.7mm]
&\cellcolor[gray]{0.8}1&\cellcolor[gray]{0.8}1&\cellcolor[gray]{0.8}1&\cellcolor[gray]{0.8}1&\cellcolor[gray]{0.8}1&\cellcolor[gray]{0.8}1&\cellcolor[gray]{0.8}1&\cellcolor[gray]{0.8}1&\cellcolor[gray]{0.8}1&\cellcolor[gray]{0.8}1&\cellcolor[gray]{0.8}1&\cellcolor[gray]{0.8}1&\cellcolor[gray]{0.8}1&\cellcolor[gray]{0.8}1&\cellcolor[gray]{0.8}1&\cellcolor[gray]{0.8}1&\cellcolor[gray]{0.8}1&\cellcolor[gray]{0.8}1&\cellcolor[gray]{0.8}1&\cellcolor[gray]{0.8}1&\cellcolor[gray]{0.8}1&\cellcolor[gray]{0.8}1&\cellcolor[gray]{0.8}1&\cellcolor[gray]{0.8}1\\
&\cellcolor[gray]{0.8}1&\cellcolor[gray]{0.8}1&\cellcolor[gray]{0.8}1&\cellcolor[gray]{0.8}1&\cellcolor[gray]{0.8}1&\cellcolor[gray]{0.8}1&\cellcolor[gray]{0.8}1&\cellcolor[gray]{0.8}1&\cellcolor[gray]{0.8}1&\cellcolor[gray]{0.8}1&\cellcolor[gray]{0.8}1&\cellcolor[gray]{0.8}1&\cellcolor[gray]{0.8}1&\cellcolor[gray]{0.8}1&\cellcolor[gray]{0.8}1&\cellcolor[gray]{0.8}1&\cellcolor[gray]{0.8}1&\cellcolor[gray]{0.8}1&\cellcolor[gray]{0.8}1&\cellcolor[gray]{0.8}1&\cellcolor[gray]{0.8}1&\cellcolor[gray]{0.8}1&\cellcolor[gray]{0.8}1&\cellcolor[gray]{0.8}1\\
&\cellcolor[gray]{0.8}1&\cellcolor[gray]{0.8}1&\cellcolor[gray]{0.8}1&\cellcolor[gray]{0.8}1&\cellcolor[gray]{0.8}1&\cellcolor[gray]{0.8}1&\cellcolor[gray]{0.8}1&\cellcolor[gray]{0.8}1&\cellcolor[gray]{0.8}1&\cellcolor[gray]{0.8}1&\cellcolor[gray]{0.8}1&\cellcolor[gray]{0.8}1&\cellcolor[gray]{0.8}1&\cellcolor[gray]{0.8}1&\cellcolor[gray]{0.8}1&\cellcolor[gray]{0.8}1&\cellcolor[gray]{0.8}1&\cellcolor[gray]{0.8}1&\cellcolor[gray]{0.8}1&\cellcolor[gray]{0.8}1&\cellcolor[gray]{0.8}1&\cellcolor[gray]{0.8}1&\cellcolor[gray]{0.8}1&\cellcolor[gray]{0.8}1\\
&\cellcolor[gray]{0.8}1&\cellcolor[gray]{0.8}1&\cellcolor[gray]{0.8}1&\cellcolor[gray]{0.8}1&\cellcolor[gray]{0.8}1&\cellcolor[gray]{0.8}1&\cellcolor[gray]{0.8}1&\cellcolor[gray]{0.8}1&\cellcolor[gray]{0.8}1&\cellcolor[gray]{0.8}1&\cellcolor[gray]{0.8}1&\cellcolor[gray]{0.8}1&\cellcolor[gray]{0.8}1&\cellcolor[gray]{0.8}1&\cellcolor[gray]{0.8}1&\cellcolor[gray]{0.8}1&\cellcolor[gray]{0.8}1&\cellcolor[gray]{0.8}1&\cellcolor[gray]{0.8}1&\cellcolor[gray]{0.8}1&\cellcolor[gray]{0.8}1&\cellcolor[gray]{0.8}1&\cellcolor[gray]{0.8}1&\cellcolor[gray]{0.8}1\\
\hline
\end{tabular}
}
\end{table}
In Table \ref{tab:1} the carrier set has four elements, $S=\{\mathtt{a,b,c,d}\}$. There are 24 different sequences, i.e. one dimensional arrangements, of the elements of $S$ and these appear in the column headings of the table. There are 16 different binary data elements for $S$ and each row of Table \ref{tab:1} gives a particular data element under the 24 different one dimensional arrangements. Now let $T:=\{S_{5},S_{10},S_{13}\}$ be the set of typical data for $S$. Let us try to arrange the elements of $S$ in a way that achieves something similar to that exemplified by the definitive player problem. We can consider which sequence, or other arrangement, of the elements of the carrier set gives the most structured, transition free, interpretation of the typical data elements. The sequence $\mathtt{acdb}$ and its reverse $\mathtt{bdca}$ satisfy this requirement since under these arrangements, for each typical data element, the zeros and ones are unmixed. In the sequel we introduce relations to show how the typical data determines the structure on the carrier set. Since this structure is given by a symmetric relation, as opposed to an antisymmetric relation in the case of total orders, the problem of whether $T$ gives $\mathtt{acdb}$ or $\mathtt{bdca}$ as the definitive arrangement of the carrier set will be solved. We begin with the following standard definitions which will be particularly useful here.
\begin{definition}
\label{def:BRRSAT}
Let $S$ be a nonempty set. A {\em{binary relation}} on $S$ is a subset $R\subseteq S^{2}$ where $S^{2}:=\{(a,b):a\in S, b\in S\}$. For $a,b\in S$ we say that $a$ is $R$-{\em{related}} to $b$, and write $aRb$, precisely when $(a,b)\in R$. We say that $R$ is:
\begin{enumerate}
\item[(i)]
{\em{reflexive}} if $(a,a)\in R$ for all $a\in S$;
\item[(ii)]
{\em{symmetric}} if for every $(a,b)\in R$ we also have $(b,a)\in R$;
\item[(iii)]
{\em{antisymmetric}} if for every pair of distinct elements $a,b\in S$ at most one of $(a,b)$ and $(b,a)$ is an element of $R$;
\item[(iv)]
{\em{transitive}} if for every triple of elements $a,b,c\in S$ with $(a,b)\in R$ and $(b,c)\in R$ we also have $(a,c)\in R$;
\item[(v)]
an {\em{equivalence relation}} if $R$ is reflexive, symmetric and transitive.
\end{enumerate}
\end{definition}
There is a strong connection between the theory of relations on a set and graph theory. In the following definition we use some graph theory terminology.
\begin{definition}
\label{def:WRD}
Let $S$ be a nonempty finite set and $R\subseteq S^{2}$ a symmetric relation on $S$. For $a,b\in S$ a {\em{walk}} from $a$ to $b$, if one exists, is a finite sequence $(k_{i})_{i\in\{1,\cdots,n\}}$, $n\in\mathbb{N}$ is odd, such that:
\begin{enumerate}
\item[(w1)]
$k_{1}=a$ and $k_{n}=b$;
\item[(w2)]
we have $k_{i}\in S$ if $i$ is odd and $k_{i}\in R$ if $i$ is even;
\item[(w3)]
for $i$ even we have $k_{i}=(k_{i-1},k_{i+1})$.
\end{enumerate}
For $a,b\in S$ let $K_{a,b}$ denote the set of all walks from $a$ to $b$. The $R$-{\em{distance}} between two elements $a,b\in S$ is
\begin{equation}
\mathrm{d}_{R}(a,b):= \left\{ \begin{array} {l@{\quad\mbox{if}\quad}l}
\min\{\frac{n-1}{2}:(k_{i})_{i\in\{1,\cdots,n\}}\in K_{a,b}\} & K_{a,b}\mbox{ is nonempty} \\
\infty & K_{a,b}=\emptyset.
\end{array} \right. 
\end{equation}
\end{definition}
\begin{lemma}
\label{lem:RD}
Since $R$ is symmetric, the $R$-distance $\mathrm{d}_{R}$ defined in Definition \ref{def:WRD} is either a metric or an extended metric on $S$. By extended metric we mean a metric that takes non-negative values on the extended real line, $[-\infty,\infty]$.
\end{lemma}
\begin{proof}
One checks the four standard metric axioms.
\end{proof}
\begin{remark}
\label{rem:SPLEX}
Let $S$ be a nonempty finite set and $n:=\# S$. Then $S^{2}$ is the equivalence relation on $S$ with just one equivalence class. Whilst the graph diagram of a graph need not be unique, by applying uniformity principals for the lengths of edges and angles between adjacent edges, many graph diagrams are unique. For example, the graph diagram of $S$ with the relation $S^{2}$ is given by the edges and vertices, nodes, of the $n-1$ dimensional regular simplex e.g. for $n=4$ the simplex is a tetrahedron.
\end{remark}
In the sequel the following metric will also be useful.
\begin{lemma}
\label{lem:SDM}
Let $S$ be a nonempty finite set and let $2^{S^{2}}$ be the set of all binary relations on $S$, noting that this is the power set of $S^{2}$. Then
\begin{equation}
\mathrm{d}_{\Delta}(R,R'):=\# (R\Delta R'),\quad R,R'\in 2^{S^{2}},
\end{equation}
is a metric on $2^{S^{2}}$ where $R\Delta R':=(R\cup R')\backslash(R\cap R')$ is the symmetric difference of $R$ and $R'$. We call $\mathrm{d}_{\Delta}$ the symmetric difference metric on $2^{S^{2}}$.
\end{lemma}
\begin{proof}
Standard for $S^{2}$ finite.
\end{proof}
The following example shows how typical data determines a structure on the carrier set.
\begin{example}
\label{exa:SPTD}
With reference to Table \ref{tab:1}, again let $T=\{S_{5},S_{10},S_{13}\}$ be the set of typical data for $S=\{\mathtt{a,b,c,d}\}$. With reference to Definition \ref{def:CSTD}, we note that each typical data element $S_{i}=\{(a,f_{i}(a)):a\in S,\quad f_{i}:S\rightarrow\{0,1\}\}$ defines an equivalence relation on $S$ of the form
\begin{equation}
R(S_{i}):=\{(a,b):a,b\in S,\quad f_{i}(a)=f_{i}(b)\}.
\end{equation}
Hence from $T$ we obtain the relation tables in Figure \ref{fig:1}. Note that for numerical cell values use $1-|f_{i}(a)-f_{i}(b)|$ for $a,b\in S$.
\begin{figure}[ht]
\centering
\includegraphics[width=0.7\textwidth]{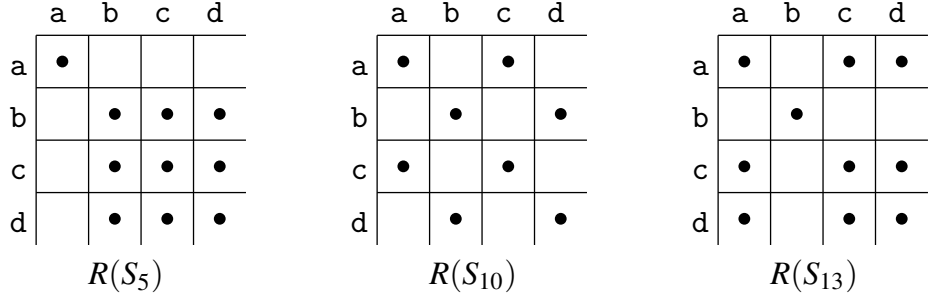}
\caption{The relation tables defined by the elements of $T$.}
\label{fig:1}
\end{figure}
Now we aggregate the relation tables in Figure \ref{fig:1} into a single weighted relation table $R_{T}$ by calculating the mean number of dots per table cell. Hence for $a,b\in S$, $R_{T}$ shows the proportion of equivalence relations defined by the elements of $T$ that have $a$ related to $b$. The table $R_{T}$ is shown in Figure \ref{fig:2}.
\begin{figure}[ht]
\centering
\includegraphics[width=0.7\textwidth]{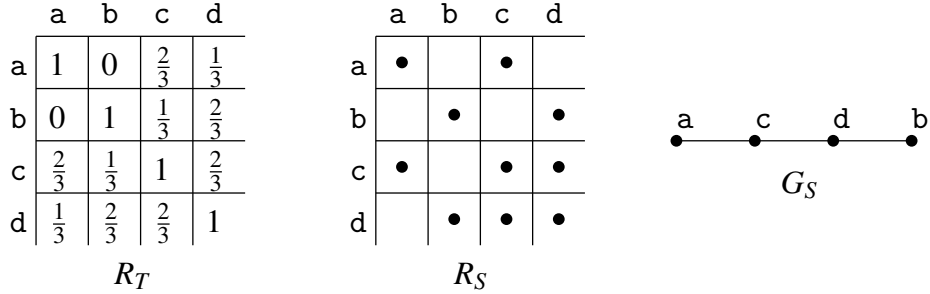}
\caption{The structure on the carrier set $S$ determined by $T$.}
\label{fig:2}
\end{figure}
Now, for a threshold value of $0.5$, we round the cell values of $R_{T}$ such that values greater than $0.5$ are rounded up to $1$ and values less than or equal to $0.5$ are rounded down to $0$. This results in the relation $R_{S}$. We note that a relation obtained in this way will always be symmetric but in general it need not be transitive. In particular $R_{S}$ is not transitive but since it is symmetric it defines a metric or an extended metric on $S$ by Lemma \ref{lem:RD}. Hence we will refer to $S$ with $R_{S}$, defined by $T$, as the {\em{carrier space}}.\\
The graph diagram of $S$ with the relation $R_{S}$ is given by $G_{S}$ in Figure \ref{fig:2}. Arguable $G_{S}$ is one dimensional and we note that it agrees with our discussion at the beginning of Section \ref{sec:1} since being a non-directed graph.\\
We note that as theory develops it might be useful to retain the weighted relation $R_{T}$ instead of only working with $R_{S}$. In particular one can obtain a hierarchy of relations from $R_{T}$ by varying the rounding threshold. However there are good reasons for choosing a rounding threshold of $0.5$. In particular $R_{S}$ is such that the mean of the distances between $R_{S}$ and the elements $R(S_{i})$ obtained from $T$ is minimized, that is
\begin{equation}
\label{equ:MEFE}
\frac{1}{\# T}\sum_{S_{i}\in T}\mathrm{d}_{\Delta}(R_{S},R(S_{i}))=\min\left\{\frac{1}{\# T}\sum_{S_{i}\in T}\mathrm{d}_{\Delta}(R,R(S_{i})):R\mbox{ is a relation on }S\right\}.
\end{equation}
In general $R_{S}$ need not be unique in this respect if the value $0.5$ appears in the relation table for $R_{T}$. We will shortly relate $R_{S}$ to something we will call float entropy which also supports a rounding threshold of $0.5$.
\end{example}
The following example uses typical data which defines a structure on the carrier set that is not one dimensional.
\begin{example}
\label{exa:SPTDTD}
With reference to Table \ref{tab:1}, let $T'=\{S_{6},S_{9},S_{16}\}$ be the set of typical data for $S':=S$ where $S$ is the carrier set of Example \ref{exa:SPTD}. Following the theory introduced in Example \ref{exa:SPTD} gives the results presented in Figure \ref{fig:3}.
\begin{figure}[ht]
\centering
\includegraphics[width=0.7\textwidth]{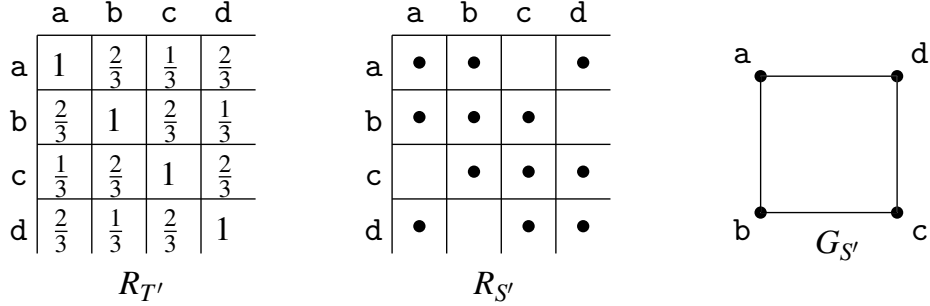}
\caption{The structure on the carrier set $S'$ determined by $T'$.}
\label{fig:3}
\end{figure}
\end{example}
\subsection{Float entropy}
\label{subsec:1FE}
In this short subsection we will discuss the notion of {\em{float entropy}}. Let $S$ be a carrier set, $T\subseteq\Omega_{S}$ the typical data of $S$ and $R$ a relation on $S$. Suppose we consider $T$ to be the set of possible messages that can be sent to a receiver. In standard information theory the receiver would also have a copy of $T$ so that sending a message only involves sending enough information to identify the intended element. Instead of this suppose that the receiver only has a copy of $S$ and $R$. For $S_{i}\in T$ if the relation $R(S_{i})$ is relatively close to $R$ with respect to $\mathrm{d}_{\Delta}$ then the number of bits that need to be sent to the receiver in order to specify $S_{i}$ will be relatively small. In this case $S_{i}$ is highly compressible, carries little information and is highly structured relative to $R$. We summarize this situation by saying that $S_{i}$  has low float entropy relative to $R$. The extreme case of minimum float entropy occurs when $R(S_{i})=R$  which is possible if $R$ is an equivalence relation. With reference to Definition \ref{def:CSTD} we can quantify float entropy relative to a given relation $R$ as follows,
\begin{equation}
\label{equ:QFE}
\mathrm{fe}(R,S_{i}):=\log_{2}(\#\{S_{j}\in\Omega_{S}:\mathrm{d}_{\Delta}(R,R(S_{j}))\leq\mathrm{d}_{\Delta}(R,R(S_{i}))\}).
\end{equation}
This is a measure in bits of the amount of information required to specify $S_{i}$ under the assumption that what is being specified ought to be highly structured relative to $R$. We can consider some values for examples \ref{exa:SPTD} and \ref{exa:SPTDTD}. Recall that in Example \ref{exa:SPTD} we have $T=\{S_{5},S_{10},S_{13}\}$ and in Example \ref{exa:SPTDTD} $T'=\{S_{6},S_{9},S_{16}\}$. For $R_{S}$ from Example \ref{exa:SPTD} we have $\mathrm{fe}(R_{S},S_{10})=1$ and $\mathrm{fe}(R_{S},S_{5})=\mathrm{fe}(R_{S},S_{13})=2.58$ to two decimal places whereas in contrast $\mathrm{fe}(R_{S},S_{9})=4$. We will denote the mean of the float entropies for the elements of $T$ with respect to $ R_{S}$ by $\mathrm{fe}(R_{S},T)$ and extend this notation to $T'$ and $R_{S'}$ from Example \ref{exa:SPTDTD} accordingly. Working to 2dp throughout gives $\mathrm{fe}(R_{S},T)=2.06$ and $\mathrm{fe}(R_{S'},T')=2.58$ whereas $\mathrm{fe}(R_{S},T')=3.55$ and $\mathrm{fe}(R_{S'},T)=3.87$. Hence we see that the relations obtained by the method shown in the examples are, relative to their respective typical data, a good choice in order to minimize the mean float entropy.\\
Now let $S$ be the set of neurons of a brain and $T$ the set of brain states where a brain state is a possible and probable aggregate state of all the brain's neurons. If we are trying to approximate $T$ then ideally $T$ will be selected such that, as a random variable restricted to $T$, the brain has a uniform distribution over $T$. Further ideally $T$ should be large enough so that the probability of the brain being in a state that is close to at least one of the elements of $T$ is high. Under these conditions we note, by Equation \ref{equ:MEFE}, that setting $R:=R_{S}$ is a good choice in order to minimize the expected float entropy.\\
In the next section we will to some extent consider the possible relevance of the theory in Section \ref{sec:1} to the brain. We will also extend the theory to what we will call objects.
\section{The brain and relations between objects defined by typical data}
\label{sec:2}
Although our theory is to be considered for typical data elements of the state of the whole brain, we begin this section by considering the relevance of the theory to the primary visual cortex, V1. Associating the retina with the unit disc of the complex plain and similarly embedding the flattened cortical sheet of V1 into the complex plain, we note that the retino-cortical mapping to V1 on a given side of the brain is approximately logarithmic and is therefore far from being an isometry, see \cite{Balasubramanian} and \cite{Schwartz}. Hence the geometry of V1 cannot account for the perceived geometry of monocular vision. Furthermore, the right side of each retina is mapped to the right side of the brain whereas the left side of each retina is mapped to the left side of the brain. Hence the signals from a given retina go to two different brain areas. Despite this the perceived geometry produces a seamless isometric version of the image on the retina. Such facts underline the need for a theory such as that initiated in this paper since we need to explain how perceived geometry is defined by the brain.\\
Let $S$ be the set of neurons in V1. Further let $a'$ and $b'$ be two distinct points that are fixed relative to the eye in a person's field of view as depicted in Figure \ref{fig:4}.
\begin{figure}[ht]
\centering
\includegraphics[width=0.5\textwidth]{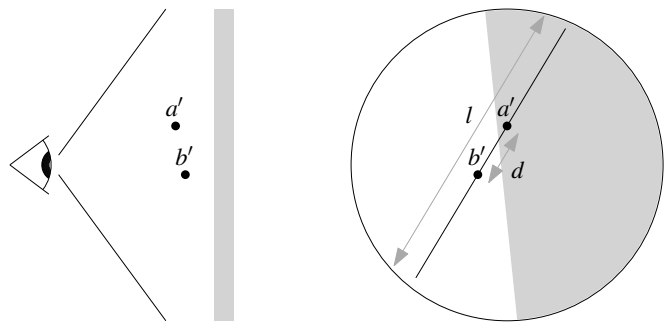}
\caption{Two fixed points in a person's field of view.}
\label{fig:4}
\end{figure}
Let $a$ be a neuron in V1 that is stimulated by the retina when there is stimulation of the retina from $a'$. Similarly let $b$ be a neuron in V1 that has the same relationship with $b'$. Consider the typical data $T$ for V1. We note that abrupt transition lines between light and dark or regions of different color are relatively sparse in the field of view. In a somewhat simplified analysis, suppose that there are usually no more than $n$ abrupt transition lines in the field of view. As depicted in Figure \ref{fig:4}, let $l$ be the length of the line through $a'$ and $b'$ crossing the field of view and $d$ the viewable distance between $a'$ and $b'$. Suppose that all $n$ transition lines intersect the line through $a'$ and $b'$. Then the probability $P_{n}$ that there is one or more transition lines between $a'$ and $b'$ is
\begin{equation}
P_{n}=1-\left(1-\frac{d}{l}\right)^{n}.
\end{equation}
We note that $\lim_{d\to 0}P_{n}=0$. Hence if $d$ is small then $a$ will be in the same state as $b$ in the majority of the typical data elements of T. On the other hand if $d$ is large then arguably $a$ and $b$ will rarely be in the same state. Therefore the relation $R_{S}$ on $S$ defined by $T$ ought to correspond well with the structure of the field of view. This claim is supported below by the results of a study using digital photographs to test how well the theory establishes relative pixel positions.\\
First though we note that evidence has been found for V1 that supports the BCM version of Hebbian theory, see \cite{BCM} and \cite{KRB}. Hebbian theory implies that if $a'$ and $b'$ are close together then stimulation of $a$ and stimulation of $b$ from within V1 ought to usually happen together. Therefore the typical data is typical of the states that V1 can internally generate by itself. Hence V1 defines $R_{S}$ and by doing so it defines the interpretation of the current state of V1. Whilst this is the case in theory further investigation is required when the full complexity of the visual system is considered.\\
Now a study was conducted using 105 digital photographs taken of everyday scenes using the same seven megapixel digital camera. A computer program centered a $5\times 5$ grid of sampling points over each photograph and recorded to which brightness class each point belonged. Here the grid points are the elements of $S$ whereas an element of $T$ is given by the values obtained for one of the photographs so that $\# T=105$. Two parameters are involved the first being the grid point spacing in pixels of adjacent grid points and the second being the number of brightness classes used. The second parameter is therefore the node repertoire and, apart from the fact that the repertoire was not restricted to two, everything proceeded as per examples \ref{exa:SPTD} and \ref{exa:SPTDTD}. Results showed that $R_{S}$ was close, with respect to $\mathrm{d}_{\Delta}$, to the relation for the grid provided that the parameters used corresponded to a point on the curve in Figure \ref{fig:5}.
\begin{figure}[ht]
\centering
\includegraphics[width=0.35\textwidth]{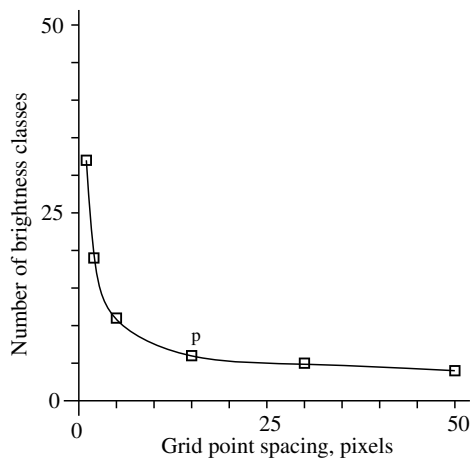}
\caption{Established parameter options.}
\label{fig:5}
\end{figure}
Now suppose we numerate the elements of $T$ from 1 to 105 and calculate $R_{S}$ after the first $n$ elements for $n\in\{1,5,10,15,\cdots,105\}$. Figure \ref{fig:6} shows how the acquired relation converged toward the relation for the grid as $n$ increased. The parameters used for Figure \ref{fig:6} are indicated by the point $\mathrm{p}$ in Figure \ref{fig:5}.
\begin{figure}[ht]
\centering
\includegraphics[width=0.4\textwidth]{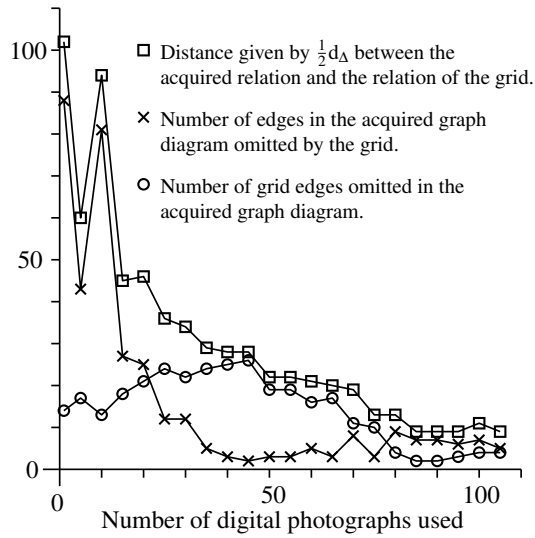}
\caption{Convergence to the relation for the grid.}
\label{fig:6}
\end{figure}
Further Figure~\ref{fig:7}, left, shows the graph diagram of the relation for the grid and, right, the edges given by the relation $R_{S}$ for $n=105$. Clearly convergence would be obtained for large enough $\# T$.
\begin{figure}[ht]
\centering
\includegraphics[width=0.5\textwidth]{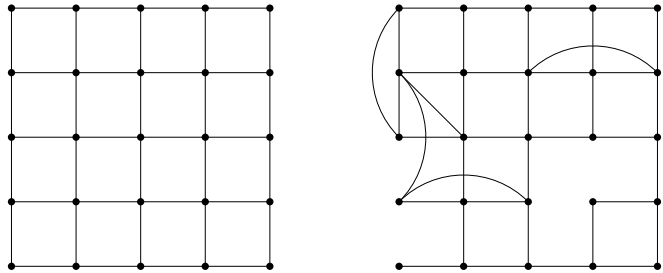}
\caption{The grid edges compared with the edges given by $R_{S}$.}
\label{fig:7}
\end{figure}
This works because whilst the content of the world around us is very varied it is nevertheless highly structured relative to the underlying geometry of the space. Brightness classes were used in the study so that the nodes, the grid points, would represent neurons in V1 that respond to rod cells in the retina. We should note that the rod cells are arranged more in the form of a hexagonal lattice than a grid. Further it would be interesting to repeat this study with each grid point split into three separate nodes giving one for each cone cell type so that $\# S=75$. The cone cells respond either to red, green or blue. The resulting relation $R_{S}$ may suggest a solution to the binding problem for color perception. Finally we should consider what might determine the repertoire of a neuron. The brain itself should define this. For example, if a small change in the output frequency of a neuron has no affect on the system then with respect to the system the neuron's state is the same. Similarly if switching over the outputs of two different neurons would have no affect on the system then with respect to the system the neurons are in the same state. This last point is just a suggestion. Note that such a definition of relative node state may result in the relation $R(S_{i})$ for $S_{i}\in T$ no longer being transitive. We will now move onto our discussion concerning objects.
\subsection{Relations between objects defined by typical data}
\label{subsec:2ROTD}
We start this subsection with a definition.
\begin{definition}
\label{def:OBJ}
Let $S$ be a nonempty finite set with typical data $T$ and the relation $R_{S}$ defined on $S$ by $T$. Let $X$ be some other finite set with $\# X\leq\# S$. We say that
\begin{equation}
X_{j}:=\{(a,x_{j}(a)):a\in X,\quad x_{j}:X\rightarrow\{0,1\}\},\mbox{ with a relation }R_{X_{j}}\mbox{ on }X_{j},
\end{equation}
is an {\em{object of}} $S$ if there is some $S_{i}\in T$, $S_{i}=\{(a,f_{i}(a)):a\in S,\quad f_{i}:S\rightarrow\{0,1\}\}$ with relation
\begin{equation}
R_{S_{i}}:=\{((a,f_{i}(a)),(b,f_{i}(b))):(a,b)\in R_{S}\},
\end{equation}
and an injective map $\Lambda_{ji}:X_{j}\rightarrow S_{i}$, given by $\Lambda_{ji}((a,x_{j}(a))):=(\lambda_{ji}(a),f_{i}(\lambda_{ji}(a)))$ where $\lambda_{ji}(a)\in S$, such that for all $(a,x_{j}(a)),(b,x_{j}(b))\in X_{j}$ we have:
\begin{enumerate}
\item[(i)]
$x_{j}(a)=f_{i}(\lambda_{ji}(a))$;
\item[(ii)]
$((a,x_{j}(a)),(b,x_{j}(b)))\in R_{X_{j}}$ if and only if $(\Lambda_{ji}((a,x_{j}(a))),\Lambda_{ji}((b,x_{j}(b))))\in R_{S_{i}}$.
\end{enumerate}
We say that the object $X_{j}$ {\em{embeds into}} $S_{i}$ and denote the set of all objects of $S$ by $\mathcal{O}$.
\end{definition}
We will now show that typical data $T$ defines a relation $R_{\mathcal{O}}$ on the set of objects of $S$ as follows. For $X_{j}\in\mathcal{O}$ let $T_{X_{j}}:=\{S_{i}:X_{j}\mbox{ embeds into }S_{i}\mbox{ where }S_{i}\in T\}$. Note, by Definition \ref{def:OBJ}, that $T_{X_{j}}$ is not empty. Now the relation $R_{\mathcal{O}}$ is given by
\begin{equation}
R_{\mathcal{O}}:=\left\{(X_{j},Y_{k}):\frac{\#(T_{X_{j}}\cap T_{Y_{k}})}{\#T_{X_{j}}}>0.5\mbox{ where }(X_{j},Y_{k})\in\mathcal{O}^{2}\right\}.
\end{equation}
We note that in general $R_{\mathcal{O}}$ need not be symmetric or transitive and that it is the relation obtained by applying a rounding threshold of $0.5$ to the weighted relation $R_{\mathcal{T}}$ given in Figure \ref{fig:8}.
\begin{figure}[ht]
\centering
\includegraphics[width=0.25\textwidth]{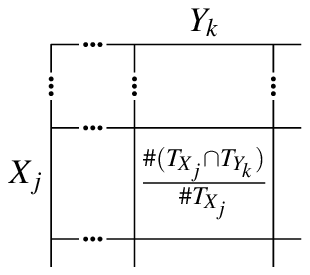}
\caption{The weighted relation table for $R_{\mathcal{T}}$ on $\mathcal{O}$ determined by $T$.}
\label{fig:8}
\end{figure}
Similar to the situation in Example \ref{exa:SPTD}, one can obtain a totally ordered hierarchy of relations on $\mathcal{O}$ by varying the rounding threshold applied to $R_{\mathcal{T}}$. Turning our attention to the topic of float entropy that we began in Subsection \ref{subsec:1FE}, we note that if the receiver not only has a copy of $S$ and $R_{S}$ but also has a copy of $R_{\mathcal{O}}$ then the elements of $T$ should be even more compressible and are even more structured relative to the relations available to the receiver. Finally we note that the theory in this paper easily generalises to cases where the neurons, or other nodes, have more than a two state repertoire, that is we can allow $f_{i}$ to take more than two values in the definition of a data element $S_{i}$ given in Definition \ref{def:CSTD}. In this case one also makes a similar adjustment to the definition of an object $X_{j}$ of $S$.
\section{Development, testing and conclusion}
\label{sec:3}
There are different ways in which this theory can be developed. From a purely theoretical perspective it is interesting to establish the range of structures that can be defined by typical data comprised of comparable nodes noting for example that functions can be defined by relations. This general theory can then be applied to any dynamical system comprised of comparable nodes, e.g. networks. More practically the wealth of established knowledge concerning brain function offers an interdisciplinary approach to theoretical development. Furthermore the theory needs to be tested. In this respect functional MRI with high spatial resolution and other brain imaging technologies could be used. For example FMRI has already been used as a way of obtaining information about the state of V1 that is sufficient for image reconstruction, see \cite{Miyawaki}. However due to spatial distortion of the retino-cortical mapping and restricted FMRI voxel resolution, and perhaps other factors, it is not possible to recognize viewed stimulus from FMRI images directly. Reconstruction often uses methods from linear mathematics and probability where knowledge of the visual stimulus used is necessary during the setup stage. Taking the elements of $S$ to be the voxels covering V1 it is interesting to know whether typical data would give rise to a geometric relationship between the voxels, differing from their FMRI image positions, such that the viewed stimulus would be recognizable from the repositioned voxels. Two methods could be tried when establishing the geometry on $S$. The first would follow the theory as presented in Section \ref{sec:2}. For the second the distances between the voxels could be obtained from the map $d:S^{2}\rightarrow[0,1]$, $d(a,b):=1-R_{T}(a,b)$. In both cases each relation $R(S_{i})$ should have numerical cell values since the similarity of voxel states can be quantified in the range 0 to 1. One type of visual stimulus to try would have a single transition line placed at random in the field of view.
\subsection{Conclusion}
\label{subsec:3CON}
We have already mentioned in Section \ref{sec:2} that the BCM version of Hebbian theory provides evidence of how the brain itself defines typical data. We mentioned the evidence in the case of the primary visual cortex V1 but there is also evidence for the relevance of BCM theory regarding the hippocampus, see \cite{DB}. In particular the typical data that V1 defines should be typical of the states induced by signals from the retina. In Section \ref{sec:2} it is shown, at least in theory and up to a good analogy using a digital camera, that for appropriate parameters such typical data defines a relation on the set of neurons of V1 that gives the perceived geometry for monocular vision. The relation is defined by the typical data by being special in the sense that it minimizes the expectation of the float entropy of the system. However our theory is intended to be applied to typical data for the whole brain so that such a relation also determines how the states of other sensory regions are perceived. For example the relation on the auditory cortex might define how we perceive the relationship between the pitches of the chromatic scale. Of course more work is required in order to determine the extent to which this theory can account for how the brain defines the various aspects of consciousness.\\
However at the higher semantic level it is fairly clear that the typical data for the brain defines relationships between objects in the way described in Subsection \ref{subsec:2ROTD}. For example a good impressionist painting provides V1 with just enough of a particular stimulus such that V1 produces the same state as that induced by a photograph of the same subject. This ability of the brain is widely known as filling-in and shows that typical data defined by the brain will determine a strong relationship between certain objects. Furthermore it is well known that certain parts of the thalamus act as a relay between different parts of the cortex including different sensory regions. This and other connections can arguably result in the brain defining typical data that determines relationships between objects arising from different sensory regions of the cortex, \cite{Shinkareva} is of relevance here. Further the states of the brain during dreaming, visualization with the eyes closed and inner sound are all instances of typical data produced by the brain itself independent of the senses at the time.\\
We will now turn our attention to what is known as the binding problem. In short the binding problem can be summarized by the following observation and question. The visual content of our conscious experience correlates with the state of the visual cortex, whereas the sound content of out conscious experience correlates with the state of the auditory cortex. How therefore can the state of two quite distinct and spatially separated brain regions give rise to a single unified conscious experience? If the theory presented in this paper is correct then the answer is quite straightforward. The content of consciousness is defined by the state of the brain interpreted in the context of the relations, such as those discussed above, defined by the brain's typical data. The typical data is determined by the brain's structure. Hence consciousness is a property of the brain as opposed to being an output of some algorithmic procedure or relying on some homunculus concept. A compact disc on its own is almost meaningless but in the context of a sufficiently large CD library it is a specific piece of music, Beethoven for example or Mozart perhaps. Similarly a brain state on its own is almost meaningless but in the context of the brain's typical data it is a moment of consciousness by which we mean the brain state with the relations defined on it by the typical data and this is for example the view of the coffee cup with the sound of the radio and the taste of the coffee all together.\\
Finally this paper if correct still leaves many questions unanswered and the lack of an attempt to answer them in the context of this initial proposition of the theory is rightful cause for some criticism. Here are a few of these questions:
\begin{enumerate}
\item[(i)]
Can the theory explain the conscious experience of the color red or does the theory need to be extended?
\item[(ii)]
What are the other relations that typical data define?
\item[(iii)]
What connections are there, if any, between our theory and the theory of consciousness as integrated information as proposed by Giulio Tononi, see \cite{Tononi},?
\item[(iv)]
Even though the neurons are an obvious candidate for the elements of the carrier set $S$ are they the right candidate?
\item[(v)]
Let $S_{i}$ be the data element for a given brain state. Is all of the relation $R_{\mathcal{O}}$ contributing to consciousness regarding $S_{i}$ or is only a subset
\begin{equation}
R_{\mathcal{O}}(S_{i}):=\{(X_{j},Y_{k}):(X_{j},Y_{k})\in R_{\mathcal{O}}\mbox{ where both } X_{j}\mbox{ and } Y_{k}\mbox{ embed in to }S_{i}\}?
\end{equation}
\item[(vi)]
Is it useful to also consider a carrier set where the elements are time dependent neurons over a short time interval or some discrete version of the same involving short finite sequences?
\end{enumerate}

{\bf\em{Acknowledgements:}}\\
I would like to thank the School of Mathematical Sciences, University of Nottingham, Nottingham, NG7 2RD, UK for providing continued access to facilities during the period after my PhD whilst I was still registered as a student and writing this paper.

\nocite{*}     
\bibliographystyle{unsrt}
\bibliography{masonComplexitybibdata}

%

\end{document}